\documentclass[aps,preprint,nofootinbib,superscriptaddress]{revtex4}%
\usepackage{hyperref}
\usepackage{amsmath}
\usepackage{amsfonts}
\usepackage{amssymb}
\usepackage{graphicx}%
\providecommand{\U}[1]{\protect\rule{.1in}{.1in}}
\newtheorem{theorem}{Theorem}

\newtheorem{lemma}[theorem]{Lemma}

\newenvironment{proof}[1][Proof]{\noindent\textbf{#1.} }{\ \rule{0.5em}{0.5em}}

\usepackage{natbib}
\usepackage{graphicx}

\usepackage{tikz-cd}

\def\be{\begin{equation}}
\def\ee{\end{equation}}

\begin{document}
\title{Halperin states can produce any filling fractions}
\author{En-Jui Kuo}
\email{kuoenjui@umd.edu}
\affiliation{Joint Quantum Institute, NIST/University of Maryland, College Park, MD 20742, USA}
\date{\today }

\begin{abstract}
We study bilayer Fractional Quantum Hall State also known as Halperin state. We prove that for any filling fractions (positive rational numbers), there are infinite solutions that imply that infinite topological states are corresponding to any given filling fraction. Our results can apply to any charge vectors.
\end{abstract}
\maketitle

\section{Introduction}

One way to describe the effective field theory of the FQHE state is to use the so-called $K$ matrix approach. One important state is called Halperin state \citep{delmastro2019symmetries, levin2007particle, seidel2008halperin}. Given any integer charge vector $t=(t_1, t_2)$, the filling fraction of such state is given by this formula \citep{tong2016lectures, polychronakos1990abelian, dunne1999aspects}:
\be \label{firsteq}
\nu=\frac{-2 l t_2 t_1+n t_1^2+m t_2^2}{m n-l^2}
\ee
where $m, n, l$ are three integers. In this case, the $K$ matrix is 
\be\label{K}
\left(\begin{array}{cc}  m & l\\ l & n \end{array}\right).
\ee
The physical meaning of $det(K)=|mn-l^2|$ corresponds to the number of nonequivalent excitation to the bilayer system. So if there are two different $K$ matrix correspond to the same filling fraction but they have different determinant, it means they are the different topological phase of matter. Here come the natural questions: fixed any filling fraction and fixed charge vector, how many topological phases of matter it has? To solve this question, we asked how many allowed integers $m, n, l$ satisfy this equation \ref{K}. We will answer these equations step by step in this paper. We use $\mathbb{Q}, \mathbb{Z}, \mathbb{N}, \mathbb{Q}_{> 0}, \mathbb{N}_{0}$ to denote set of rational numbers, integers, positive integers, positive rational numbers, nonnegative integers respectively. $p \perp q$ means $p$ and $q$ are coprime to each other. We are ready to state our two main theorems.

\begin{theorem}\label{th1}
Given $(t_1, t_2) \in \mathbb{N}^2 \setminus (0,0)$ and any rational filling fractions $\nu \in \mathbb{Q}_{> 0}$, there is always a non negative integer solution $m, n, l$ such that $\nu=\frac{p}{q}$ for given $\nu$ which satisfy $mn-l^2 > 0$.
This is equivalent to, the following set $A_{t_1, t_2}=\mathbb{Q}_{> 0}$ which is defined to be.
\be
A_{t_1, t_2}=\{ \frac{-2 l t_2 t_1+n t_1^2+m t_2^2}{m n-l^2} | m \geq 1, n \geq 1, l \geq 0, (m, n, l) \in \mathbb{N}_{0}^3, mn-l^2 \geq 1\}.
\ee 
\end{theorem}

\begin{theorem}
Moreover, given any $C >0$. There always exists some $(m, n, l) \in \mathbb{N}_{0}^3$ produce the filling fraction $\nu$ such that $mn-l^2>C$. This fact shows that the number of solution $(m, n, l)$ is \textbf{infinite} which imply: 
\textbf{for any rational filling fractions, there exists infinite type of phase of matter and they all give the same filling fractions.} We can collect the determinant of the solutions come from $\nu$. We donated this set as $sol_{\nu, t_1, t_2}$.
\be
sol_{\nu, t_1, t_2}=\{ mn-l^2| m \geq 1, n \geq 1, l \geq 0, (m, n, l)  \in \mathbb{N}_{0}^3, mn-l^2 \geq 1, \frac{-2 l t_2 t_1+n t_1^2+m t_2^2}{m n-l^2}=\nu \}.
\ee

The above statement is the same as $sol_{\nu, t_1, t_2}$ is unbounded or $\sup sol_{\nu, t_1, t_2} = \infty$ where sup is defined to be the supremum of a set.
\end{theorem}

Given $(t_1, t_2) \in \mathbb{N}^2 \setminus (0,0)$, let $\nu=\frac{p}{q}, p \perp q$, where $p, q$ are positive integers. One can write eq \ref{firsteq} into this Diophantine equation:
\be
p(mn-l^2)-q(-2 l t_2 t_1+n t_1^2+m t_2^2)=0.
\ee
Our theorems are equivalent to there are always infinite integer solutions where $m>0, n>0, l\geq0, mn-l^2>0$. In general, for any multivariable polynomial equation, whether it has an integer solution is \textbf{undeciable} \citep{koenigsmann2014undecidability}. Namely, no algorithm or systematic approach can be used to find whether it has an integer solution. Besides, it is also very difficult to find how many solutions given an equation. There are many solved or unsolved problems in the number theory. For examples, the famous solved examples including Fermat Last Theorem \citep{darmon1995fermat}, open problems including Erdos–Straus conjecture \citep{terzi1971conjecture} and \citep{wooley2000sums} which is called the sum of three cubes, it is an open problem to characterize the numbers that can be expressed as a sum of three cubes of integers, allowing both positive and negative cubes in the sum. There is an enormous conjecture of this kind of problem. After reviewing these three problems, one may think they are all dealing with the multivariable polynomial which degree higher than 2. So one may think the quadratic integer equation is simple to understand, however, this is not the case. For example $x^2+y^2=z^2$ has infinite nonzero solutions while $x^2+y^2=3z^2$ is not. Another famous example is negative Pell's equation \citep{barbeau2006pell} which is quadratic but still unclear whether it has a solution in some cases although it only has two variables. More deep results related to Quadratic Diophantine equations can be seen here \citep{watson1960quadratic}. We list some interesting Quadratic Diophantine equations in Appendix c to illustrate the nontriviality of Quadratic Diophantine equations and our result.
Recently, The relation between arithmetic functions and the Abelian Chern Simon theory with different symmetries can be found in \citep{delmastro2019symmetries}. One final remark is that the number of different topological excitations of the given $K$ matrix is $|det(K)|$. But in this paper, we only consider $det(K)>0.$

The paper is organized as follows. In sections II and III, we briefly prove these two theorems in the special case. In section IV, we prove our general results. Our conclusion and partial generalization are included in section V, VI.

\section{First Special Case $t=(1, 0)$}

As the toy example, we try to prove the $t=(1, 0)$ case. In other words, we want to prove $A_{1, 0} =\mathbb{Q}_{>0}$ and $\sup sol_{\nu, 1, 0} = \infty$. The derivation is simple. In this case, the filling fraction is $\nu=\frac{p}{q}=\frac{n}{mn-l^2}$. We have to show given any filling fractions, there exists infinite solutions which generate this filling fraction.\\
\begin{proof} First we can choose large enough integer $m$ such that $mp-q \geq 0$ and we choose $K$ matrix to be 
\be\label{case 1}
K=\left(\begin{array}{cc}  m & pm-q\\ pm-q & p(pm-q) \end{array}\right).
\ee
One can easily check this gives us correct filling fractions $\frac{p}{q}$ and the determinant $det K= pm(pm-q)-(pm-q)^2=pmq-q^2=q(pm-q)$ which can be arbitrary large by choosing large enough integer $m$. So we prove $A_{1, 0} =\mathbb{Q}_{>0}$ and $\sup sol_{\nu, 1, 0} = \infty$. We present another interesting way to prove there are infinite solutions. One can notice that if $(m, n, l)$ is a solution, then $(m, \alpha^2 n, \alpha l)$ is also a solution where $\alpha \in \mathbb{N}$. So one solution can generate infinite solutions. Unfortunately, this type of linear integer transformation is not easily to be found in general charge vector case.
\end{proof}

\section{Second special case $t=(1, 1)$}

In this section, we fixed charge vector $t=(1, 1)$. The filling fractions $\nu$ of the Halperin state (double layer fractional quantum hall ) is given by this formula:
\be \label{main}
\frac{m+n-2l}{mn-l^2}=\nu
\ee
where $m>0,n>0,l \geq 0.$ and $(m, n, l)$ must be integers. We will prove that just by restricting in this state or 2 by 2 $K$ matrix, one can already generate any rational number filling fraction $\nu \in \mathbb{Q}_{>0}$. Any filling fractions can be made by cleverly choosing three non negative integers $(m, n, l)$ such that they satisfy eq (\ref{main}) and these constraints: $mn-l^2>0$ and $ m>0, n>0, l \geq 0$. Besides, for any rational filling fractions, there are infinite solutions. In this case $t=(1, 1)$, we prove nontrivial results and specific construction for the solutions.

Let us show some examples before going to the details. For charge vector $t=(1, 1)$ and rational filling fraction $\nu$, we denote the solution as $(m, n, l, mn-l^2)$. One can easily write a computer program to search for solutions. We present some solutions of $\nu=\frac{2}{3}, \frac{13}{17}, \frac{16}{17}$. We will show that given any filling fractions, there are always infinite solutions. Now we can state this fact and give a rigorous proof.

\begin{tabular}{ |p{3cm}|p{3cm}|p{3cm}|  }
\hline
\multicolumn{3}{|c|}{Partial solution List $(m, n, l, mn-l^2)$} \\
\hline
$\nu=\frac{2}{3}$ & $\nu=\frac{13}{17}$ & $\nu=\frac{16}{17}$ \\
\hline
$(2, 6, 3, 3)$ & $(2, 21, 5, 17)$ & $(2, 272, 17, 255)$ \\
$(2, 14, 4, 12)$ &  $(2, 66, 8, 68)$   & $(5, 159, 26, 119)$ \\
$(2, 62, 7, 75)$ & $(2, 137, 11, 153)$ & $(12, 194, 47, 119)$\\
$(2, 86, 8, 108)$  & $(5, 68, 17, 51)$ & $(47, 497, 152, 255)$ \\
$(3, 15, 6, 9)$ & $(9, 42, 19, 17)$ & $(54, 320, 131,  119)$  \\
$(3, 39, 9, 36)$ & $(21, 66, 37, 17)$ & $(75, 369, 166,  119)$  \\
$(3, 123, 15, 144)$ & $(20, 113, 47, 51)$ & $(77, 587, 212, 255)$ \\
\hline
\end{tabular}

\begin{theorem}
Given $0<\nu=\frac{p}{q}, p, q \in \mathbb{N}, p \perp q.$ We have $A_{1, 1}=\mathbb{Q}_{>0}$ and $\sup sol_{\nu, 1, 1} = \infty$.
\end{theorem}
\begin{proof}
In order to prove this, we need to solve the equation:
\be \label{Main}
q(m+n-2l)-p(mn-l^2)=0
\ee
First, since $p,  q$ are coprime to each other, there must exist some integers $k$ such that $n=2l+pk-m.$ So the eq \ref{Main} is a quadratic equation for $m$.
\be
p \left(-k m p+k q+l^2-2 l m+m^2\right)=0.
\ee
We can solve $m$.
\be \label{m}
m=\frac{1}{2} \left(\pm \sqrt{k} \sqrt{k p^2+4 l p-4 q}+k p+2 l\right).
\ee
In order to make $m$ to be an integer, the $k(k p^2+4 l p-4 q)$ should be a perfect square. Let's say $\rho \in \mathbb{N}$. We have
\be
k(k p^2+4 l p-4 q)=\rho^2.
\ee
Again, this is a quadratic equation for $\rho$. We can solve $k$ in terms of $\rho$.
\be\label{k}
k=\frac{\pm \sqrt{4 (q-l p)^2+p^2 \rho ^2}-2 l p+2 q}{p^2}
\ee
The same reason is that $4 (q-l p)^2+p^2 \rho ^2=x^2$ where $x$ is an integer. One can solve $\rho$. We get
\be
\rho =\pm \frac{\sqrt{x^2-4 (q-l p)^2}}{p}
\ee
In order to let $\rho$ to be an integer. We define a positive integer $s=\frac{\sqrt{x^2-4 (q-l p)^2}}{p}$. Up to now, we change $(m, n, l)$ three variables to $q, x ,s$ where s is also an integer.
\be \label{key}
-p^2 s^2-4 (lp-q)^2+x^2=0.
\ee
This relation eq \ref{key} is the fundamental equation we need to solve. Up to now, we can solve this equation by \textbf{Pythagorean Triples}. From the \textbf{Pythagorean Triples}. There are infinite positive integer solutions such that $x^2+y^2=z^2$. The proof is really simple to notice that $(a^2-b^2)^2+(2ab)^2=(a^2+b^2)^2$ for any real numbers $a, b$. Here comes the main part of the proof. For more details of Pythagorean Triples, one can see Appendix b for more details.

We separate this problem into two cases: in the first cases, assuming $ps, 2(lp-q), x$ can be coprime. Notice since $p, q$ are fixed, we first need to prove there is always a solution $l, s, x$. First, one can always pick $l$ large enough such that $lp-q>0$ and we hope we can find there exists $a, b, a> b$ such that $a^2-b^2=ps, lp-q=ab$. This can be done by taking $a,b, a>b$ such that $a-b=pt$, $t$ is an positive integer. This is not true in general $(p, q)$. One can see Lemma in Appendix a. If $q$ is quadratic residue over $p$ and then this can be done. In this case, we can see
$s=t(a+b), x=a^2+b^2=b^2+(b+pt)^2$. We can solve $k$ by plug $\rho=\pm s$ into eq \ref{k}. There are two solutions of $k$. 
\be
k=\frac{-2 l p+2 q \pm x}{p^2}= t^2, -\frac{(2 b+p t)^2}{p^2}.
\ee
We can choose $k=t^2$, obvious this is integer. One can go back and solve $m, n$ by eq \ref{m}. 
\be
m=\frac{1}{2} (k p+2 l \pm s)=t (b+p t)+l, l-bt.
\ee
We can choose large enough integer $t$ such that $m=t(b+pt)+l$ is a positive integer. In this case, the corresponding $n=l-bt$. One may worry $n$ is negative. We may notice that $l=\frac{a b+q}{p}=\frac{(b+pt) b+q}{p}$. So $n=\frac{b^2+q}{p}>0$. Finally we can check the determinant $mn-l^2=q t^2>0$. This term can be arbitrary large since $t$ can be arbitrary large. So we finish the existence proof and show that there are infinite solutions. We give an examples here. Suppose $p=3, q=5$, we need to find $a^2-b^2=3s, 3l-5=ab, a,b>0, a>b.$ We can choose $a=4, b=1, t=1$.
We can find $(m, n, l)=(7, 3, 2).$ One can easily check in this case $\nu=\frac{3}{5}.$

In the other case, suppose $q$ is not quadratic residue of modulo $p$. One can still find infinite pairs of many nonzero integers $u, a, b, l$ such that $(a-b)=pt, pl-q=uab$. This proof is really simple, first we choose $l_0 \in \mathbb{N}$ such that $pl_0-q=u_0>0$. Then we have $p(l_0+u_0t)-q=u_0+u_0pt=u_0(1+pt).$ Suppose we choose $a=1+pt$ and $b=1$. In this case, we have $m=l_0+t u (p t+2), n=l_0, l=l_0+tu$. This gives us $\nu=\frac{p}{q}$. Arbitrary $u, t$ can give us infinite solutions. Determinant in this special case is $uqt^2$. We also give an example here. Suppose $\nu=\frac{5}{8}$, we then choose $u_0=2, l=2+2t, m=2 t (5 t+2)+2, n=2$. It is direct to check that:
\be
\frac{5}{8}=\frac{m+n-2l}{mn-l^2}=\frac{-2 (2 t+2)+2 t (5 t+2)+4}{16 t^2}.
\ee
Indeed, this parameterized solutions produce the filling fraction and the determinant of K matrix is $16t^2$ which can be arbitrary large. In this section, we established the following: $A_{1, 1}=\mathbb{Q}_{>0}, \sup sol_{\nu, 1, 1} = \infty.$
\end{proof}

One may ask whether we need $l$ as a free variable to generate all the filling fractions and generate infinite solutions. The answer is yes by the following two theorems. We will give a concise proof in the case of the general vector in the next section.
\begin{theorem}
Given $l \in \mathbb{N}$, there are always finitely many solutions $m, n$ such that $\nu=\frac{p}{q}$ for any $\nu$.
\end{theorem}
\begin{proof} The reason is simple since only finite $s, x$ can satisfy eq \ref{key}.
\end{proof}

\begin{theorem}
Given $l \in \mathbb{N}_{0}$, there \textbf{exist} some filling fractions $0<\nu=\frac{p}{q}<1$ such that there is \textbf{no} solution $m, n$. 
\end{theorem}
\begin{proof}

If $l=0$, it is simple to see that the following set
$\{ \frac{m+n}{mn}=\frac{1}{m}+\frac{1}{n} | (m, n) \in \mathbb{N}^2 \}$ can not generate rational number $\nu$ where $\frac{5}{6}<\nu<1$.

For $l>0$, there exist two positive numbers $p, q$ such that $q-lp=1$ and the absolute value of $p, q$ can be arbitrary large. Suppose we choose large enough $p$, we can find the only solution for eq \ref{key} is $s=0, x=\pm 2$. However, once we put $x=\pm 2, s=0$ into $k=\frac{\pm \sqrt{4 (q-l p)^2+p^2 \rho ^2}-2 l p+2 q}{p^2}=\frac{4}{p^2}$. Then $k$ can not be integers as $p$ large enough.
\end{proof}

We give an example when $l=1$, without loss of generosity, we can assume that $m>n \geq 1$.
\be
\nu=\frac{m+n-2}{mn-1}.
\ee
since $(p, q)=1$, one can assume that $n=kp-m+2, k \in \mathbb{N}$. One can solve $m$ in terms of $k$. We get
\be
m=\frac{1}{2}(\pm \sqrt{k} \sqrt{kp^2+4p-4q}+kp+2).
\ee
In order $m$ to be positive integer, $\sqrt{k} \sqrt{kp^2+4p-4q}$ should be a square of some integers $\rho$. $k \left(k p^2+4 p-4 q\right)=\rho ^2$. One can solve $k$. We can get $k=\frac{\pm \sqrt{p^2 \rho ^2+4 (p-q)^2}+2 p-2 q}{p^2}$ . In order to make this $k$ as an integer, one need to make $p^2 \rho ^2+4 (p-q)^2$ as a square of some integers $x$.
\be
p^2 \rho ^2+4 (p-q)^2=x^2
\ee
Up to now, we can already see only finite solutions when fixed $(p, q)$.  $x^2-p^2 \rho ^2=(x-p\rho)(x+p\rho)=4(p-q)^2$ since the right hand side contains only \textbf{finite} number of factorization. So one can solve $x, \rho$ case by case if there is a solution. 

For example, $(p, q)=(3, 4)$. There is no solution to it since $\rho=\pm \frac{1}{4} \sqrt{x^2-4}$. The only integer $x=2$ and in this case: $k=0$ which leads to $m=1, n=1$ which is not admissible. One can convince that there are infinite $(p, q)$ such that there are no solutions. Other examples are $(p, q)=(5, 7), (p, q)=(4, 5)$. They do not have solution if one fixes $l=1$. 

Before we close this section, we provide two integer examples. First, we use $\nu=1$ as illustration: if $\nu= 1$. The solution is quite simple. We want to solve $\frac{m+n-2l}{mn-l^2}=1$ which implies $(m-1)(n-1)=(l-1)^2$. It is obviously that there are infinite solutions and determinant can be arbitrary large. We demonstrated another example here. For $\nu=7$, we can choose $m=1+6t(7t+2), n=1, l=1+6t$. One can easily check every $t>0$ giving $\nu=7$.

\section{General charge $t$ vectors}

Now, we are able to prove the general charge vector $t_1, t_2$. We assume that $(t_1, t_2)$ are integer vectors which components are greater or equal to zero and one of them must be positive. Our theorems are still valid. In this case, the filling fraction is
\be
\nu=\frac{-2 l t_2 t_1+n t_1^2+m t_2^2}{m n-l^2}.
\ee
We want to show the set $A_{t_1, t_2}=\mathbb{Q}_{>0}$.
We separate our proof into two steps.
\paragraph{Case 1: One of $(t_1, t_2)$ is zero, i.e $A_{t_1, 0}= \mathbb{Q}_{>0}.$}
Without loss of generosity, we can assume $t_1 \neq 0$. This is the same as $t=(1, 0)$ case by modifying any filling fraction $\frac{p}{q} \to \frac{p}{t_1^2q}$.
\paragraph{Case 2: None of $(t_1, t_2)$ is zero,  i.e $A_{t_1,  t_2}= \mathbb{Q}_{>0}.$}

Without loss of generosity, we can assume they are coprime to each other. In this case, one can show that all the theorems hold to produce all the filling fraction. To prove this fact, we first provide a simple lemma.
\begin{lemma}
If one can produce all the positive integers then one can produce any positive rationals. Namely if $\mathbb{N} \subset A_{t_1, t_2}$, we have $A_{t_1, t_2}=\mathbb{Q}_{>0}.$ 
\end{lemma}
\begin{proof} Suppose we know that $(m, n, l)$ can produce an integer $p$,
then $(qm, qn, ql)$ can produce a filling fraction $\frac{p}{q}$ where $q$ is a positive integer. 
\end{proof}

Now, we are able to give a general proof. We will use the same approach as $t=(1, 1)$ case but restrict ourselves on $p \in \mathbb{N}$. We want to solve the following equations:
\be\label{pp}
p \left(l^2-m n\right)-2 l t_2 t_1+m t_2^2+n t_1^2=0
\ee
The proof is separated into two steps. We first show one can produce integers $p$ where $p \geq 2$. 
We make a change of variable. Let $k=-2 l t_2 t_1+m t_2^2+n t_1^2$, we can write $n$ in terms of $k$ and put it back in eq \ref{pp}. We need to solve the following equation:
\be
\frac{m p \left(-k-2 l t_1 t_2+m t_2^2\right)}{t_1^2}+k+l^2 p=0.
\ee
We can solve $l$. We get 
\be
l= \frac{\pm\frac{\sqrt{k m p^2-k p t_1^2}}{t_1}+\frac{m p t_2}{t_1}}{p}.
\ee
We can choose plus sign for the $l$. To make $l$ to be an integer, a safe choice is to let $m=t_1^2$ and then we plug it back to $l$. We have $l=\sqrt{\frac{k (p-1)}{p}}+t_1 t_2.$ We can choose $k=(p-1) p x^2$ where $x$ is an positive integer. We then solve $n$.
We find 
\be
n=\frac{(p-1) p x^2}{t_1^2}+\frac{2 (p-1) t_2 x}{t_1}+3 t_2^2.
\ee
For the generic $p$, if we choose $x=t_1 \beta$ where $\beta$ is a positive integer. After doing this, one can go back and express $n, l$ in terms of $\beta$.
To summarize, one can check if we choose $K$ matrix to be:
\be
K=\left(
\begin{array}{cc}
 t_1^2 & (p-1) t_1 \beta +t_1 t_2 \\
 (p-1) t_1 \beta+t_1 t_2 & (p-1) p \beta ^2+2 (p-1) t_2 \beta+t_2^2\\
\end{array}
\right).
\ee
One can check that the filling fraction is $p$. The determinant $K$ is $\beta ^2 (p-1) t_1^2$ which can be arbitrary large. By choosing such integer $\beta$, all entries of the $K$ matrix are positive. We finish this part of the proof.

Now we do the $p=1$ case. Let $k=-2 l t_2 t_1+m t_2^2+n t_1^2$ and we write $n$ in terms of other variables and we plug $n$ into the eq \ref{pp}. We get
\be\label{p1}
\frac{m \left(-k-2 l t_1 t_2+m t_2^2\right)}{t_1^2}+k+l^2=0.
\ee
In order to make this equation holds for general $t_1, t_2$. A safe choice is to choose $m=2t_1^2$, notice that $m=t_1^2$ is not a suitable choice since $k$ will vanished in the equation. If we admit that $m=2t_1^2$ and solve $l$ in eq \ref{p1}. One get $l=\sqrt{k}+2 t_1 t_2$ so $k$ should be a square number.
We let $k=u^2$. We can go back and solve $n$. We get $n=\frac{2 t_1 t_2 \left(t_1 t_2+u^2\right)+u^2}{t_1^2}$. A save choice is to choose $u=\beta t_1$ where $\beta$ is a positive integer. To sum up, one can choose the corresponding $K$ matrix to be:
\be
K=\left(
\begin{array}{cc}
 2 t_1^2 & t_1 \beta+2 t_1 t_2 \\
 t_1 \beta+2 t_1 t_2 & \beta^2+2 t_2 \beta+2 t_2^2 \\
\end{array}
\right).
\ee
One can plug them in the filling fraction formula and check:
\be
\nu=\frac{-2 l t_2 t_1+n t_1^2+m t_2^2}{m n-l^2}=1.
\ee
The determinant are $t_{1}^{2} \beta^2$ which can be arbitrary large. From this proof, we can generate all the integers, this shows that we can produce all the filling fractions from the lemma. Due to infinite parameters $\beta$, we have infinite solution and determinant can be arbitrary large implying
$\sup sol_{\nu, t_1, t_2} = \infty$. These results are stronger results comparing to the previous special two $t$ cases. By choosing large enough $\beta$, all the elements in the $K$ matrix can be made positive. However, in the previous case, we give more detail constructions. 

Before we close this section, we want to prove another interesting theorem. We already show one can generate any filling fractions if one allows $(m, n, l).$ How about if one fixes the $l \in \mathbb{N}_{0}$? Can the following $K$ matrix generate all the filling fractions? 
\be
K=\left(
\begin{array}{cc}
 m & l_0 \\
 l_0 & n \\
\end{array}
\right).
\ee
We will show any finite restrictions of $l_0$ can not generate all the filling fractions. To be more precisely,
we already know the following:
\be
\mathbb{Q}_{>0}=A_{t_1, t_2}= \bigcup _{l_0\in \mathbb{N}_0} A_{t_1, t_2}^{l_0}
\ee
where 
\be
A_{t_1, t_2}^{l_0}=\{ \frac{-2 l_0 t_2 t_1+n t_1^2+m t_2^2}{m n-l_0^2} |(m, n)\in \mathbb{N}^2, mn-l_0^2 \geq 1\}
\ee
and $\bigsqcup$ means disjoint union.
We have the following theorem.
\begin{theorem}
$\forall l_0 \in \mathbb{N}_0$
\be
A_{t_1, t_2}^{l_0} \subsetneq \mathbb{Q}_{>0}
\ee
\end{theorem}

\begin{proof}
In order to proof this, we only need to show $A_{t_1, t_2}^{l_0}$ is a bounded set. If $A_{t_1, t_2}^{l_0}$ is a bounded set, it can not be $\mathbb{Q}_{>0}.$ We write $A_{t_1, t_2}^{l_0}$ as 
\begin{align}
A_{t_1, t_2}^{l_0}= &\{ \frac{-2 l_0 t_2 t_1+n t_1^2+m t_2^2}{m n-l_0^2} |(m, n)\in \mathbb{N}^2, mn>2l_0^2, mn-l_0^2 \geq 1\} \cup \nonumber \\
 &\{ \frac{-2 l_0 t_2 t_1+n t_1^2+m t_2^2}{m n-l_0^2} |(m, n)\in \mathbb{N}^2, mn \leq 2l_0^2, mn-l_0^2 \geq 1\}.
 \end{align}
The second set is finite since only finite $m, n$ such that $mn \leq 2l_0^2.$ So the second set is bounded due to finiteness.
The first set is bounded since
\be
\frac{-2 l_0 t_2 t_1+n t_1^2+m t_2^2}{m n-l_0^2} \leq \frac{n t_1^2+m t_2^2}{m n-l_0^2}\leq \frac{n t_1^2+m t_2^2}{\frac{mn}{2}} \leq 2t_1^2+2t_2^2
\ee
Notice that we have used $mn>2l_0^2$ and $m, n \geq 1$. So two bounded sets union is also a bounded set. So we show $A_{t_1, t_2}^{l_0}$ is bounded. This means for given $(t_1, t_2, l_0)$. There exists a maximum filling fraction $\nu$ depends $(t_1, t_2, l_0)$ one can produce. Actually, one can get more general results. We basically use the fact that the union of finitely many bounded set is bounded. So we have 
\be
\bigcup_{\text{finite } l_0} A_{t_1, t_2}^{l_0}\subsetneq \mathbb{Q}_{>0}.
\ee
This theorem indicates the non trivial fact of the original theorem. One can generate any filling fractions if one allows $(m, n, l).$ However, any finite allowed $l$ can not generate any filling fractions.
\end{proof}

\section{FURTHER CLASSIFICATION RESULTS}

In this section, we make more restrictions on the $K$ matrix. We define Bosonic sets and fermonic sets as follows:
\be
B_{t_1, t_2}=\{ \nu(t_1, t_2, m, n, l) | m \geq 1, n \geq 1, l \geq 0, (m, n, l) \in \mathbb{N}_{0}^3, mn-l^2 \geq 1, m=0(mod2)=n. \}
\ee
\be
F_{t_1, t_2}=\{ \nu(t_1, t_2, m, n, l) | m \geq 1, n \geq 1, l \geq 0, (m, n, l) \in \mathbb{N}_{0}^3, mn-l^2 \geq 1, m=1(mod2)=n. \}
\ee
where $\nu(t_1, t_2, m, n, l) \equiv \frac{-2 l t_2 t_1+m t_2^2+n t_1^2}{m n-l^2}$.
Obviously, We have $B_{t_1, t_2} \subset A_{t_1, t_2}, F_{t_1, t_2} \subset A_{t_1, t_2}$. Now the question is: \textbf{Is $B_{t_1, t_2}, F_{t_1, t_2}$ can generate all the fractions?}. i.e is $B_{t_1, t_2}=\mathbb{Q}_{>0}$ or $F_{t_1, t_2}=\mathbb{Q}_{>0}$. We start to answer these questions in this section. Naively we have this inclusion diagram as a set. Now inclusion relation represents the subset relation.
In the previous section, we already prove $A_{t_1, t_2}=\mathbb{Q}_{>0}$. We want to ask two question marks in the diagram.

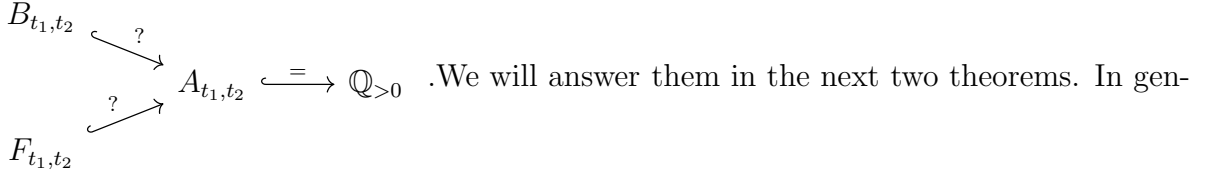
\begin{tikzcd}[row sep=tiny]
B_{t_1,t_2}\arrow[dr," ?", hook]  &   & \\
                             &  A_{t_1,t_2} \arrow[r, "=",hook] & \mathbb{Q}_{>0}\\ 
F_{t_1,t_2}\arrow[ur," ?", hook]  &  &
\end{tikzcd}
.We will answer them in the next two theorems. In general $t_1, t_2$, we have the following diagram. 
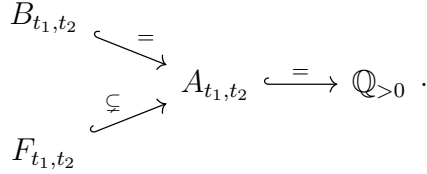
\begin{tikzcd}[row sep=tiny]
B_{t_1,t_2}\arrow[dr,"=", hook]  &   & \\
                             &  A_{t_1,t_2} \arrow[r, "=",hook] & \mathbb{Q}_{>0}\\ 
F_{t_1,t_2}\arrow[ur,"\subsetneq ", hook]  &  &
\end{tikzcd}.

\begin{theorem}
$B_{t_1, t_2}=\mathbb{Q}_{>0}$ and every fraction has infinite bosonic $K$ matrices.
\end{theorem}
\begin{proof}
For the first part, suppose one of charge vector is zero. Without loss of generosity, we can assume $t_2=0$. We have ansatz eq \ref{case 1}. Suppose filling fraction is $\frac{p}{q}.$ If $q$ is even, then $p$ must be odd. We can choose $m$ to be even and $n=pm-q$ is even. If $q$ is odd, and $p$ is even, then we choose $m$ to be even. If $p$ and $q$ are both odds, pick an positive integer $\alpha$ which is even and positive. We build $K$ matrix as follows:
\be\label{case1}
K=\left(\begin{array}{cc}  m & p\alpha m-1\\ p\alpha m-1 & p\alpha(p\alpha m-1) \end{array}\right).
\ee
This $K$ matrix generates filling fraction $p\alpha.$ So we can times the matrix by a factor $\alpha q$,.
\be\label{case1}
K=\alpha q\left(\begin{array}{cc}  m & p\alpha m-1\\ p\alpha m-1 & p\alpha(p\alpha m-1) \end{array}\right).
\ee
This $K$ matrix produce the filling fraction $\frac{p}{q}$ and all elements are even. This proves $B_{1, 0}=\mathbb{Q}_{>0}.$ The similar techniques can be used to prove the general charge vector cases. 
\end{proof}

\begin{theorem}
$F_{1, 0}$ can \textbf{not} generate all the fractions. i.e $F_{1, 0} \subsetneq \mathbb{Q}_{>0}.$
\end{theorem}
\begin{proof}
The key is that any integer times even number is even. In this case $\nu=\frac{n}{mn-l^2}=\frac{p}{q}$ and $p, q$ are coprime to each other. If we choose $p$ is even and $q$ is odd, then $n$ must be even. If this is not the case, then $p, q$ will not be coprime. So we proved $F_{1, 0}$ can not generate all the filling fractions. In the general charge vector case, whether $F_{t_1, t_2}$ can generate all the fractions will depend on the structure of $t_1, t_2$.
\end{proof}

\section{Conclusion}

In this project, we prove that given any rational number $\nu>0$ and any integers charge vector $t_1, t_2$ at least one of them are nonzero integers. There exist infinite pairs of three integers $m \geq 1, n\geq 1, l \geq 0$ such that $\nu=\frac{-2 l t_2 t_1+m t_2^2+n t_1^2}{m n-l^2}$. i.e, infinite allowed $K$ matrix. Moreover, given any $C>0$, one can find a solution which determinant $mn-l^2>C$. The physical meaning of these theorems is that there are infinite different types of the topological phase of matter which can produce for any given rational filling fractions $\nu>0$. There are some other theorems mentioned in the text. We summarize our theorems. Given $(t_1, t_2) \in \mathbb{N}^2 \setminus (0,0)$ and any rational filling fractions $\nu \in \mathbb{Q}_{> 0}$, we have the following:

\begin{center}
    \begin{tabular}{| l | l | }
    \hline
     Theorems &  Meaning \\ \hline
    $A_{t_1, t_2}=\mathbb{Q}_{>0}    \nonumber $ & Any filling fractions can be produced. \\ \hline
    $sup sol_{\nu, t_1, t_2}=\infty$ & There are always infinite solutions for a given filling fraction. \\ \hline
    $ \bigcup_{\text{finite } l_0} A_{t_1, t_2}^{l_0}\subsetneq \mathbb{Q}_{>0}$ & Any finite collections of $l_0$ can not generate any filling fractions.  \\ \hline
  $B_{t_1, t_2}=\mathbb{Q}_{>0} $ & There always exist infinite Bosonic matrices for a given filling fraction.  \\ \hline
  $F_{t_1, t_2} \subsetneq \mathbb{Q}_{>0}$ & Fermonic matrices are not always exist for a given filling fraction.  \\ \hline
    \end{tabular}
\end{center}

\begin{acknowledgments}
The author would like to thank Sun-Ting Tsai, Yixu Wang, Jerry YC Hu, Mia Chen for helpful comments on our draft and support this project. The author is supported by the JQI (Joint Quantum Institute) Research Assistantship and Professor Mohammad Hafezi.
\end{acknowledgments}

\section{Appendix}
\paragraph{Quadratic residue}
In number theory, an integer $q$ is called a quadratic residue modulo $n$ if it is congruent to a perfect square modulo $n$; i.e., if there exists an integer $x$ such that:
\be
{\displaystyle x^{2}\equiv q{\pmod {n}}.}
\ee
Otherwise, $q$ is called a quadratic nonresidue modulo n. We prove a theorem related to the quadratic residue.

\begin{lemma}
Given $p, q \in \mathbb{N}, (p, q)=1$ and if  $q$ or $-q$ is a quadratic residue modulo $p$. There are always infinitely many  $a>0, b>0, l>0, (a, b ,l )\in \mathbb{N}$ such that $p|(a^2-b^2)$ and $lp-q=ab$.
\end{lemma}

\begin{proof}
Suppose $-q$ is a quadratic residue modulo $p$. Then there exist some $l_0$ such that $pl_0-q=h^2$ where $h$ is an integer.
Then one can easily see if we choose $p(l_0+h)-q=h^2+ph=ab$, and then choose $a=h+p, b=h$. We find that $p|a-b$. Also there are infinite and arbitrary large $l_0$ from the theory of Quadratic residue. This proof our lemma and the rest follows as main proof. Conversely, if there exist $a, b, l, t$ such that $ab=pl-q, a-b=pt$, then $q$ is a Quadratic residue modulo $p$ since $pl-q=b(b+pt)=b^2+bpt$.

On the other hand, Suppose $q$ is a quadratic residue modulo $p$. Then there exists some $l_0$ such that $pl_0+q=h^2$ where $h$ is a positive integer. We can write it as $-pl_0-q=-h^2$ the flip $l_0 \to -l_0$. We have $pl_0-q=-h^2$. We can choose large enough $n$ such that $p(l_0+nh)-q=-h^2+pnh=ab$ is large positive number. We can choose $a=-h+pnh$ and $b=h$. We have $a+b=pnh$ then $p|(a+b).$  We finish the proof. We give a example which $q$ or $-q$ is not quadratic residue modulo $p$. We choose $p=5, q=3$. It is easy to see $[3]_5$ and $[-3]_5=[2]_5 $ are not quadratic residue by computing Legendre Symbol \citep{SystemModeler}.
\end{proof}

\paragraph{Pythagorean triple}
A Pythagorean triple consists of three positive integers $a$, $b$, and $c$, such that $a^2 + b^2 = c^2$. Such a triple is commonly written $(a, b, c)$, and a well-known example is $(3, 4, 5)$. If $(a, b, c)$ is a Pythagorean triple, then so is $(ka, kb, kc)$ for any positive integer $k$. A primitive Pythagorean triple is one in which  $a$, $b$, and $c$ are coprime.

Euclid's formula is a fundamental formula for generating Pythagorean triples given an arbitrary pair of integers $m$ and $n$ with $m > n > 0$. The formula states that the integers
\be
{\displaystyle a=m^{2}-n^{2},\ \,b=2mn,\ \,c=m^{2}+n^{2}}.
\ee
form a Pythagorean triple. The triple generated by Euclid's formula is primitive if and only if $m$ and $n$ are coprime and not both odd. When both m and n are odd, then $a$, $b$, and $c$ will be even, and the triple will not be primitive.
Despite generating all primitive triples, Euclid's formula does not produce all triples—for example, $(9, 12, 15)$ cannot be generated using integer $m$ and $n$. This can be remedied by inserting an additional parameter $k$ to the formula. The following will generate \textbf{all} Pythagorean triples uniquely:
\be
a=k\cdot (m^{2}-n^{2}),\ \,b=k\cdot (2mn),\ \,c=k\cdot (m^{2}+n^{2}).
\ee
where $m, n$, and $k$ are positive integers with $m > n$, and with $m$ and $n$ coprime and not both odd.

\paragraph{Some interesting Quadratic Diophantine equations}
The purpose of this section is to convince the reader that Quadratic Diophantine equations are already difficult to be solved in general. Namely, how many integer solutions of given equations is difficult to know. We present four examples here to show the nontriviality of our results.

We first choose this equation to illustrate its quite nontrivial in two variables. We want to find integer solutions $(x, y) \in \mathbb{Z}^2$ of the following equation:
\be
3x^2+2=y^2.
\ee
There is no solution to this equation. The proof is simple. We modulo $3$ on both side, we get $2=y^2 (mod 3)$. But one can check every integer square module $3$ can only be $1 or 0$. So there is no solution for this equation. This is the same technique we show in Appendix a since $2$ is not the quadratic residue modulo $3$. The second example is:
\be
7x^2+2=y^3.
\ee
If we modulo $7$ to both sides of the equation. We have $y^3=2(mod 7)$. However, any integer cubes modulo $7$ can not be $2$. We list $y^3(mod7)$ in the table.
\begin{displaymath}
\begin{array}{|c c|c|}
y(mod7) & y^3 (mod7) \\ % Use & to separate the columns
\hline % Put a horizontal line between the table header and the rest.
0 & 0  \\
1 & 1 \\
2 & 2^3=8=1(mod7) \\
3 & 3^3=27=6(mod7) \\
4 & 4^3=64=1(mod7) \\
5 & 5^3=125=6(mod7) \\
6 & 6^3=216=6(mod7) \\
\end{array}
\end{displaymath}

The three variable quadratic equations are much more tricky. We only list one example here. Suppose we want to find nontrivial solution $(x, y, z) \in \mathbb{Z}^3$ and $(x, y, z) \neq (0, 0, 0)$ of this equation.
\be
3x^2-7y^2-11z^2=0.
\ee
There is no nontrivial solution for this equation. WLOG, one can assume that no common factors for $(x, y, z)$. We modulo $7$ both sides and get $3x^2=4z^2 (mod 7)$. We will get $x=z=0(mod7)$. However, if we put them back in the original equation, we have $y=0(mod7)$. This contradicts there are no common factors for $(x, y, z)$. Amazingly, if we change the original equation to $3x^2-7y^2-17z^2=0.$ There are nontrivial solutions like $(x, y, z)=(5, 1, 2), (8, 5, 1)$.
This indicates by changing only one parameter in the Diophantine equation will lead to completely different results. One can see Hasse-Minkowski Theorem or Hasse principle \citep{gamzon2006hasse} for general homogeneous Quadratic Diophantine equations. However, our equations are not homogeneous, so we can not use this theorem. These examples show that our results are nontrivial.

\bibliographystyle{plain}
\bibliography{references}

\begin{thebibliography}{10}

\bibitem{barbeau2006pell}
Edward~J Barbeau.
\newblock {\em Pell’s equation}.
\newblock Springer Science \& Business Media, 2006.

\bibitem{darmon1995fermat}
Henri Darmon, Fred Diamond, and Richard Taylor.
\newblock Fermat’s last theorem.
\newblock {\em Current developments in mathematics}, 1995(1):1--154, 1995.

\bibitem{delmastro2019symmetries}
Diego Delmastro and Jaume Gomis.
\newblock Symmetries of abelian chern-simons theories and arithmetic.
\newblock {\em arXiv preprint arXiv:1904.12884}, 2019.

\bibitem{dunne1999aspects}
Gerald~V Dunne.
\newblock Aspects of chern-simons theory.
\newblock In {\em Aspects topologiques de la physique en basse dimension.
  Topological aspects of low dimensional systems}, pages 177--263. Springer,
  1999.

\bibitem{gamzon2006hasse}
Adam Gamzon.
\newblock The hasse-minkowski theorem.
\newblock 2006.

\bibitem{SystemModeler}
Wolfram~Research{,} Inc.
\newblock System{M}odeler, {V}ersion 12.0.
\newblock Champaign, IL, 2019.

\bibitem{koenigsmann2014undecidability}
Jochen Koenigsmann.
\newblock Undecidability in number theory.
\newblock In {\em Model theory in algebra, analysis and arithmetic}, pages
  159--195. Springer, 2014.

\bibitem{levin2007particle}
Michael Levin, Bertrand~I Halperin, and Bernd Rosenow.
\newblock Particle-hole symmetry and the pfaffian state.
\newblock {\em Physical review letters}, 99(23):236806, 2007.

\bibitem{polychronakos1990abelian}
Alexios~P Polychronakos.
\newblock Abelian chern-simons theories in 2+ 1 dimensions.
\newblock {\em Annals of Physics}, 203(2):231--254, 1990.

\bibitem{seidel2008halperin}
Alexander Seidel and Kun Yang.
\newblock Halperin bilayer quantum hall states on thin cylinders.
\newblock {\em Physical review letters}, 101(3):036804, 2008.

\bibitem{terzi1971conjecture}
DG~Terzi.
\newblock On a conjecture by erd{\"o}s-straus.
\newblock {\em BIT Numerical Mathematics}, 11(2):212--216, 1971.

\bibitem{tong2016lectures}
David Tong.
\newblock Lectures on the quantum hall effect.
\newblock {\em arXiv preprint arXiv:1606.06687}, 2016.

\bibitem{watson1960quadratic}
George~Leo Watson.
\newblock Quadratic diophantine equations.
\newblock {\em Philosophical Transactions of the Royal Society of London.
  Series A, Mathematical and Physical Sciences}, 253(1026):227--254, 1960.

\bibitem{wooley2000sums}
Trevor~D Wooley.
\newblock Sums of three cubes.
\newblock {\em Mathematika}, 47(1-2):53--61, 2000.

\end{thebibliography}
\end{document}